\documentclass[journal,twocolumn,10pt,letter]{IEEEtran}

\usepackage{graphics}
\usepackage{colortbl}
\usepackage{multicol}
\usepackage{lipsum}
\usepackage{color}
\usepackage[cmex10]{amsmath}
\usepackage{amsthm}
\usepackage{amssymb}
\usepackage{mathrsfs}
\usepackage{mathtools}
\usepackage{amsbsy}
\usepackage[colorlinks=true,bookmarks=false,citecolor=blue,urlcolor=blue]{hyperref}
\usepackage{xcolor}
\usepackage{mathrsfs}
\usepackage{setspace}
\usepackage{float}
\usepackage{graphicx}
\usepackage{pstool}
\usepackage{lettrine}
\usepackage{newclude}
\usepackage[normalem]{ulem}
\usepackage{latexsym}
\usepackage{algpseudocode}
\usepackage{algorithm,algpseudocode}
\usepackage{algorithmicx}

\usepackage{multirow}
\usepackage{amsmath}
\usepackage{amsmath,bm}
\usepackage{wasysym}
\usepackage{mathrsfs}
\usepackage{amsmath,cite,amsfonts,amssymb,color}
\allowdisplaybreaks
\ifCLASSOPTIONcompsoc
\usepackage[caption=false,font=normalsize,labelfon
t=sf,textfont=sf]{subfig}
\else
\usepackage[caption=false,font=footnotesize]{subfig}
\fi

\allowdisplaybreaks

\newtheorem{theorem}{Theorem}

\newtheorem{proposition}{Proposition}

\graphicspath{{figures/}}
\allowdisplaybreaks

\newcounter{mytempeqncnt}
\newcounter{mytempeqncnt2}

\usepackage{balance}

\begin{document}

\title{Designing Cost- and Energy-Efficient Cell-Free Massive MIMO Network with Fiber and FSO Fronthaul Links}

\author{\IEEEauthorblockN{Pouya Agheli\thanks{P. Agheli, M. J. Emadi, and H. Beyranvand are with the Department of Electrical Engineering, Amirkabir University of Technology (Tehran Polytechnic), Tehran, Iran (E-mails: \{pouya.agheli, mj.emadi, beyranvand\}@aut.ac.ir).},
Mohammad Javad Emadi, and Hamzeh Beyranvand
}} 

\maketitle
\begin{abstract}

The emerging cell-free massive multiple-input multiple-output (CF-mMIMO) is a promising scheme to tackle the capacity crunch in wireless networks. Designing the optimal fronthaul network in the CF-mMIMIO is of utmost importance to deploy a cost- and energy-efficient network. In this paper, we present a framework to optimally design the fronthaul network of CF-mMIMO utilizing optical fiber and free space optical (FSO) technologies. We study an uplink data transmission of the CF-mMIMO network wherein each of the distributed access points (APs) is connected to a central processing unit (CPU) through a capacity-limited fronthaul, which could be the optical fiber or FSO. Herein, we have derived achievable rates and studied the network's energy efficiency in the presence of power consumption models at the APs and fronthaul links. Although an optical fiber link has a larger capacity, it consumes less power and has a higher deployment cost than that of an FSO link. For a given total number of APs, the optimal number of optical fiber and FSO links and the optimal capacity coefficient for the optical fibers are derived to maximize the system's performance. Finally, the network's performance is investigated through numerical results to highlight the effects of different types of optical fronthaul links.
\end{abstract}
\begin{IEEEkeywords}
Cell-free massive multiple-input multiple-output, capacity-limited optical fronthaul, optical fiber link, free space optical link, achievable uplink rate, cost efficiency, and energy efficiency
\end{IEEEkeywords}

\IEEEpeerreviewmaketitle

\section{Introduction}
One of the main aspects of the emerging generations of wireless networks is to support the growing number of mobile users with high data rate demands. The massive multiple-input multiple-output (mMIMO) provides remarkable improvements in spectral- and energy-efficiency, and more interestingly, it accompanies near-optimal linear processing due to the weak law of large numbers \cite{marzetta2010noncooperative}. Besides, thanks to the uplink and downlink channel reciprocity in the time-division duplex (TDD) transmission scheme, it is only required to know the channel state information (CSI) at the base station, and the user only needs to know the statistical average of the effective channel. Thus, the pilot transmission overhead in the downlink is relaxed \cite{ marzetta2016fundamentals}. Due to the mMIMO scheme's importance, various variations of the scheme are also presented in the literature \cite{marzetta2016fundamentals,khormuji2015generalized,ngo2017cell,kabiri2017optimal}. Especially, cell-free massive MIMO (CF-mMIMO) is also introduced to not only capture the gain of mMIMO scheme but also serve users with almost the same quality \cite{ngo2017cell}.

In the cell-free massive MIMO network, all access points (APs) concurrently serve all user equipments (UEs) in the same time-frequency resource. Thus, it provides higher data rates and wider coverage area since there are closer APs to the served UEs. For instance, the  $95\%$-likely network throughput of the cell-free massive MIMO downlink is about seven times higher than that of the small-cell without shadow fading correlation \cite{ngo2017cell}. The CF-mMIMO networks are analysed in the literature from different perspectives. The power control mechanisms for data and pilot transmissions are respectively considered in \cite{ngo2017cell} and \cite{mai2018pilot}. The user-centric approach is also introduced and analysed in \cite{buzzi2017cell}. The user-centric massive MIMO network offers higher per-user data rates in comparison to the CF-mMIMO approach with less backhaul overhead demand \cite{buzzi2017user,buzzi2019user}. Moreover, \cite{ngo2017total} and \cite{yang2018energy} have studied energy efficiency and total power consumption models at APs and backhaul links in CF-mMIMO network, where it is assumed that backhaul links are \emph{ideal}. In contrast, spectral- and energy-efficiency expressions are also analysed for \emph{limited-capacity} fronthaul CF-mMIMO networks in \cite{masoumi2019performance,bashar2019max}, and the quantization effects on the performance are studied.

In practice, the optical fiber is primarily used for fronthaul and backhaul links due to its high data capacity and low path-loss, which comes at the disadvantage of high deployment cost. In contrast, although the free space optical (FSO) technology still offers large enough data capacity with much lower deployment cost and fast setup time, it has drawbacks such as pointing error, atmospheric-dependent channel quality, and needing line-of-sight (LOS) connection  \cite{hassan2017statistical,khalighi2014survey,kaushal2016optical}. To address the outage problem of the FSO links in adverse atmospheric conditions, the combined radio frequency (RF) and FSO solution is employed for fronthaul and backhaul links, namely hybrid RF/FSO and relay-based cooperation \cite{douik2016hybrid,touati2016effects,chen2016multiuser,usman2014practical,zhang2009soft,jamali2016link,najafi2017optimal}. Moreover, \cite{ahmed2018c} has studied the joint deployment of RF and FSO fronthaul links in the uplink of a cloud-radio access network (C-RAN), and \cite{najafi2017c} presents a C-RAN network with RF multiple-access links and hybrid RF/FSO fronthaul links wherein RF-based fronthaul and multiple-access links exploit the same frequency band with the optimized time-division mechanism. Besides, CF-mMIMO network for an indoor visible light communications (VLC) is investigated in \cite{kizilirmak2017centralized,beysens2018densevlc} without considering the limited-capacity fronthaul/backhaul links effects. Recently, the uplink cell-free and user-centric mMIMO networks with radio/FSO fronthaul and multi-core fiber (MCF) backhaul links have been investigated in \cite{agheli2020performance}, where the optimal power allocation and adaptive fronthaul assignment have been proposed. On the other hand, meeting the increase of the wireless networks' densification and required bandwidth, the third generation partnership project (3GPP) has evolved an integrated access and backhaul (IAB) architecture at millimeter waves (mmWaves) for the fifth generation of cellular networks (5G). To this end, the same infrastructure and resources are employed for both the access and backhaul links \cite{saha2018bandwidth,polese2018distributed}. In this case, the power and spectrum allocation, distributed path selection, and bandwidth partitioning in IAB-enabled small-cell networks have been studied in \cite{polese2020integrated,guerrero2020integrated,saha2019millimeter}.

In this paper, we analyse a cost- and energy-efficient cell-free massive MIMO network with capacity-limited optical fronthaul links connecting distributed APs to the central processing unit (CPU). It is assumed that each of the fronthaul links is deployed based on optical fiber or free space optical technology according to optimal fronthaul allocation mechanism. The main contributions of the paper are summarized as follows.
\begin{itemize}
    \item Closed-form uplink achievable data rates have been derived by employing maximum-ratio combining (MRC) and use-and-then-forget (UatF) techniques. 
    \item To optimally design the fronthaul network, we derive closed-form spectral- and energy efficiency expressions of the network by considering the fiber and the FSO deployment cost functions, power consumption models at APs, and consumed power for data transmission over the fronthaul links.
    \item We derive closed-form expressions for the optimal capacity coefficient of the fibers and the optimal number of fiber and FSO fronthaul links for a given number of APs by maximizing the network's energy efficiency.
    \item The network's performance is compared for different deployment sets of the fibers' capacity coefficient and the number of fiber-based fronthauls from different viewpoints through extensive numerical results.
\end{itemize}

\textit{Organization}: In Section \ref{Sec:Sec2}, cell-free massive MIMO network with optical fronthaul links is introduced. Uplink achievable data rates and energy efficiency are derived in Section \ref{Sec:Sec3}. Section \ref{Sec:Sec4} presents the fronthaul link allocation optimization problem, and numerical results and discussions are represented in Section \ref{Sec:Sec5}. Finally, the paper is concluded in Section \ref{Sec:Sec6}.

\textit{Notation}: $x \! \in \! \mathbb{C}^{n \times 1}$ denotes a vector in an $n$-dimensional complex space, $\mathbb{E}\{\cdot\}$ is the expectation operator, and $[\,\cdot\,]^T$ stands for the transpose. Also, $y \!\sim\! \mathcal{N}(m,\sigma^2)$ and $z \!\sim\! \mathcal{CN}(m,\sigma^2)$ respectively denote real-valued and complex symmetric Gaussian random variables (RVs) with mean $m$ and variance $\sigma^2$.


\section{System Model} \label{Sec:Sec2}
We assume a wireless cell-free massive MIMO (CF-mMIMO) network to simultaneously serve ${K}$ UEs by $M$ distributed APs in the same time-frequency resources. The system model is depicted in Fig.~\ref{Fig:Fig.1}. All UEs and APs are assumed to be single-antenna \footnote{It can be extended to a system with multi-antenna APs under more complex mathematical computations.} and distributed over a large area. Each AP is connected to a CPU via a fronthaul link. Two types of fronthaul link are considered, FSO and optical fiber (OF), where $M_{\normalfont\text{FSO}}$ of the APs are connected to the CPU via the FSO links and the rest, i.e., $M_{\normalfont\text{OF}}=M-M_{\normalfont\text{FSO}}$, are connected via the optical fiber links. 
Moreover, it is assumed that the information capacity for OF ($C_{\normalfont\text{OF}}$) is $N$ times of that of FSO ($C_{\normalfont\text{FSO}}$), i.e., $C_{\normalfont\text{OF}}=N C_{\normalfont\text{FSO}}$ for $N \geq 1$, to consider a higher capacity for the fiber links at the cost of more deployment costs than that of the FSO ones. 
Besides, we assume that the CPU has full-CSI, and we only focus on the uplink data transmission throughout the paper to analyse the performance of the system and the effects of fronthaul types.

\begin{figure}[!t]
\centering
\pstool[scale=0.58]{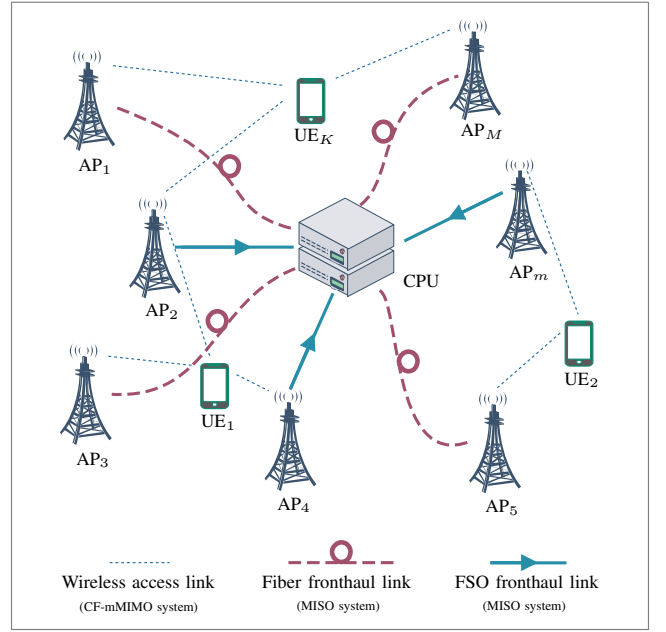}{
\psfrag{UE1}{\hspace{0.01cm}\scriptsize $\text{UE}_1$}
\psfrag{UE2}{\hspace{0.01cm}\scriptsize $\text{UE}_2$}
\psfrag{UEK}{\hspace{-0.12cm} \scriptsize $\text{UE}_K$}
\psfrag{AP1}{\hspace{0.2cm}\scriptsize $\text{AP}_1$}
\psfrag{AP2}{\hspace{0.1cm}\scriptsize $\text{AP}_2$}
\psfrag{AP3}{\hspace{0.1cm}\scriptsize $\text{AP}_3$}
\psfrag{AP4}{\hspace{0.1cm}\scriptsize $\text{AP}_4$}
\psfrag{AP5}{\hspace{0.1cm}\scriptsize $\text{AP}_5$}
\psfrag{AP6}{\hspace{0.1cm}\scriptsize $\text{AP}_6$}
\psfrag{APm}{\hspace{0.1cm}\scriptsize $\text{AP}_m$}
\psfrag{APM}{\hspace{0.1cm}\scriptsize $\text{AP}_M$}
\psfrag{CPU}{\hspace{0.1cm}\scriptsize CPU}
\psfrag{WirelessAccessLink}{\hspace{-0.05cm}\scriptsize Wireless access link}
\psfrag{OpticalFronthaulLink}{\hspace{-0.1cm} \scriptsize Fiber fronthaul link}
\psfrag{FSOFronthaulLink}{\hspace{-0.23cm} \scriptsize FSO fronthaul link}
\psfrag{(CFmMIMO)}{\hspace{-0.26cm} \tiny(CF-mMIMO system) }
\psfrag{(MISO1)}{\hspace{-0.18cm}\tiny(MISO system)}
\psfrag{(MISO2)}{\hspace{-0.18cm}\tiny(MISO system)}
}
\caption{CF-mMIMO wireless network with fiber and FSO fronthaul links.}
\label{Fig:Fig.1}
\end{figure}

\subsection{Channel Model}   
The wireless channel between the $k$th UE and the $m$th AP follows the flat-fading model and is given by 
\begin{equation} \label{Eq:Eq1}
g_{mk}=\sqrt{{\beta}_{mk}} h_{mk},
\end{equation}
where ${\beta}_{mk}$ and $h_{mk}$ respectively represent large- and small-scale fadings \cite{ngo2017cell,agheli2020performance,masoumi2019performance}. The small-scale fading coefficients $h_{mk}$ for $m=1, 2, ..., M$ and $k=1, 2, ..., K$,  are independent and identically distributed (i.i.d.) zero-mean complex Gaussian random variables with unit variance, i.e. $h_{mk}\sim \mathcal{CN}(0,1)$. 
\subsection{Uplink Data Transmission}
In the uplink data transmission, all $K$ UEs simultaneously transmit their data to APs on the same time-frequency resource element. Thus, the $m$th AP receives the following superimposed signal
\begin{equation} \label{Eq:Eq7}
y_{m}=\sqrt{\rho_u}\sum\limits_{k=1}^{K} \sqrt{\eta_k}g_{mk}q_k+\omega_{m},
\end{equation}
where $\rho_u$ represents the maximum transmit power of each user, $q_k\sim \mathcal{CN}(0,1)$ is the data symbol of the $k$th user with power control coefficient $\eta_k \in [0,1]$, and $\omega_{m}\sim \mathcal{CN}(0,\delta^2_m)$ denotes additive Gaussian noise at $m$th AP.
\subsection{Rate Distortion Theory}
To perfectly represent an arbitrary continuous random variable, one needs an \textit{infinite} number of bits. Thus, describing it with a \textit{finite} number of bits, results in \textit{distortion}. This concept is analysed in the well-known rate distortion theory \cite{cover2012elements}, which is also called the vector quantization with \textit{large enough} dimension. Let us assume an i.i.d. random source $X \sim P_X(x)$ with zero-mean and bounded variance $\mathbb{E}\lbrace |X|^2\rbrace  =\sigma^2$. 
Thus, the rate-distortion problem is to represent the $n$-length i.i.d. source $\boldsymbol{X}$ with $nR$ bits such that average distortion becomes less than $D$, i.e. $\mathbb{E}\big\lbrace d(\boldsymbol{X},\boldsymbol{\hat{X}})\big\rbrace\leq D$, where $\boldsymbol{\hat{X}}(i)$ for $i \in \big\lbrace 1, 2, ..., 2^{nR} \big\rbrace$ represents the quantized versions of the random source. 
Thus, the rate distortion function for the $X$ with square-error distortion and large enough $n$ is defined as follows \cite{cover2012elements}
\begin{equation} \label{Eq:Eq8}
R^*(D)={\underset{P_X(\hat{x}\mid x):\:   \mathbb{E}(|\hat{x}-x|^2) \leq D}{\text{min}}} I(\hat{X};X) ~~~\text{[bits per symbol]}.
\end{equation}
By use of a test channel $\hat{X}=X+Z$, where $Z \sim \mathcal{CN}(0,D) $ denotes the quantization noise which is independent of  $X$, the rate distortion function $R(D) = \log_2\big(1+\frac{\sigma^2}{D}\big)$ is achieved \cite[Chapter 10]{cover2012elements} and \cite{masoumi2019performance}\footnote{There is another variation of the test channel $X=\hat{X}+Z$ in which $Z \sim \mathcal{CN}(0,D) $ represents the quantization noise which is independent of  $\hat{X}$. In this case, the rate distortion function is obtained as $R(D) = \log_2\big({\sigma^2}/{D}\big)$ which is smaller than $\log_2\big(1+{\sigma^2}/{D}\big)$. In the paper, based on the mathematical simplicity, we have used one of the two mentioned test channels.}.

On the other hand, we then transmit the $nR$-bit quantization index over a fronthaul link with channel capacity $C$ [bits per channel use] without error. Thus, to minimize the quantization noise variance of $D$, we must use the maximum capacity of the fronthaul channel. Therefore,  $\log_2\big(1+\frac{\sigma^2}{D^*}\big)=C$, and $D^*=\frac{\sigma^2}{2^C-1}$.





\section{Performance Analysis} \label{Sec:Sec3}
In the following section, we present the uplink achievable data rates. Thus, after receiving the uplink data signals at APs, each access point quantizes and forwards the signal to the CPU over its allocated fronthaul link. Then, users' data recovery is performed at the CPU, and the achievable rates are derived.
\subsection{Quantization}
After receiving the superimposed signal $y_m$ given in (\ref{Eq:Eq7}) at   $m$th AP, by use of the rate distortion theory, the AP compresses the received signal to $\hat{y}_m$. Thus, we have
\begin{eqnarray} \label{Eq:Eq10}\nonumber
    \hat{y}_m \!\!&=&\!\! y_m+n_m,\\
    \!\!&=&\!\! \sqrt{\rho_u}\sum\limits_{k=1}^{K} \sqrt{\eta_k}g_{mk}q_k+\omega_{m} + n_m,
\end{eqnarray}
where $n_m \sim \mathcal{CN}(0,D_m)$ represents the quantization noise at $m$th AP which is independent of $y_m$.
To perfectly send the quantization index over the FO fronthaul link or FSO one with capacity $C_m$, we have  $D_m=\frac{\mathbb{E}\lbrace  |y_m|^2\rbrace}{2^{C_m}-1}$. Thus, the CPU receives $\hat{y}_m$ perfectly from the $m$th AP.
\subsection{User Data Detection}
The CPU receives the compressed signals given in (\ref{Eq:Eq10}) from all the APs, and applies the MRC technique to recover the data symbols of all the users. To recover the information symbol of $k$th UE, i.e., $q_k$, we have
\begin{align}\label{Eq:Eq12} \nonumber
r_k &= \sum\limits_{m=1}^{M}\hat{y}_m g_{mk}^*\\ \nonumber
 &=\sum\limits_{m=1}^{M}\Big(\sqrt{\rho_u}\sum\limits_{k^\prime=1}^{K} \sqrt{\eta_{k^\prime}}g_{m{k^\prime}}q_{k^\prime}+\omega_{m} + n_m  \Big) g_{mk}^*\\ \nonumber
 &=\sqrt{\rho_u}\sum\limits_{m=1}^{M}\sum\limits_{k^\prime=1}^{K} \sqrt{\eta_{k^\prime}}g_{m{k^\prime}} g_{mk}^* q_{k^\prime} + \sum\limits_{m=1}^{M}\big(\omega_{m} + n_m \big)g_{mk}^*\\ \nonumber
 &= \underbrace{\sqrt{\rho_u}\sum\limits_{m=1}^{M}\sqrt{\eta_{k}}|{g_{mk}}|^2 q_{k}}_{A_k}\\ 
 &~+ \underbrace{\sqrt{\rho_u}\sum\limits_{m=1}^{M}\sum\limits_{\substack{k^\prime=1 \\ k^\prime\neq k}}^{K} \sqrt{\eta_{k^\prime}}g_{m{k^\prime}}  g_{mk}^* q_{k^\prime} + \sum\limits_{m=1}^{M}\big(\omega_{m} + n_m \big) g_{mk}^*}_{B_k},
\end{align}
where $A_k$ represents the desired signal, and $B_k$ denotes the sum of multi-user interference and noise.

\subsection{Uplink Achievable Data Rates}
By applying the well-known use-and-then-forget (UatF) technique \cite{marzetta2016fundamentals,masoumi2019performance}, we can relax the availability of the full instantaneous CSI at the users. Instead, it is sufficient to know only the statistical average of the effective channel. Thus, (\ref{Eq:Eq12}) is rewritten as \cite{agheli2020performance}
\begin{equation} \label{Eq:Eq13}
r_k = \text{DS}_k\cdot q_k+\underbrace{\text{BU}_k \cdot q_k + \sum\limits_{\substack{k^\prime=1 \\ k^\prime\neq k}}^{K} \text{I}_{k k^\prime}\cdot q_{k^\prime} + \upsilon_k}_{\text{effective noise}},
\end{equation}
where
\begin{itemize}
    \item $\text{DS}_k$ is the desired signal of the $k$th user
    \begin{equation} \label{Eq:Eq14}
        \text{DS}_k=\sqrt{\rho_u}\:\mathbb{E}\bigg\lbrace  \sum\limits_{m=1}^{M}\sqrt{\eta_{k}}|g_{mk}|^2 \bigg\rbrace  ,
    \end{equation}
    \item $\text{BU}_k$ denotes the beamforming uncertainty of the $k$th user due to the statistical knowledge of CSI
    \begin{equation} \label{Eq:Eq15}
        \text{BU}_k=\sqrt{\rho_u}\:\Bigg(\sum\limits_{m=1}^{M}\sqrt{\eta_{k}}|g_{mk}|^2 - \mathbb{E}\bigg\lbrace  \sum\limits_{m=1}^{M}\sqrt{\eta_{k}}|g_{mk}|^2\bigg\rbrace  \Bigg),
    \end{equation}
    \item $\text{I}_{k k^\prime}$ represents the inter-user interference from $k^\prime$th users
    \begin{equation} \label{Eq:Eq16}
        \text{I}_{k k^\prime}=\sqrt{\rho_u}\sum\limits_{m=1}^{M} \sqrt{\eta_{k^\prime}}g_{m{k^\prime}} g_{mk}^*,
    \end{equation}
    \item $\upsilon_k$ is the composition of both the additive Gaussian and the quantization noises
     \begin{equation} \label{Eq:Eq17}
        \upsilon_k= \sum\limits_{m=1}^{M}\big(\omega_{m} + n_m \big)g_{mk}^*.
     \end{equation}
\end{itemize}
It can be shown that all terms given in  (\ref{Eq:Eq13}) are mutually uncorrelated. From the information theoretic point of view to analyse the worst-case scenario, we assume that all the terms, except the desired term, are modeled by an equivalent Gaussian random variables with the same variances. Thus, the uplink achievable rate $[\text{bps/Hz}]$ of the $k$th user is presented as
\begin{equation} \label{Eq:Eq18}
R_{u,k}=\log_2\big(1+\gamma_k\big),
\end{equation}
where $\gamma_k$ represents SINR of the $k$th user, which is given by
\begin{equation} \label{Eq:Eq19}
\gamma_k = \dfrac{|\text{DS}_k|^2}{\mathbb{E}\big\lbrace  |\text{BU}_k|^2\big\rbrace  +\sum\limits_{\substack{k^\prime=1 \\ k^\prime\neq k}}^{K} \mathbb{E}\big\lbrace  |\text{I}_{k k^\prime}|^2\big\rbrace   +\mathbb{E}\big\lbrace  |\upsilon_k|^2\big\rbrace}.
\end{equation}
In the following Theorem, the SINR of the $k$th user is computed.

\begin{theorem} The $k$th user has following SINR 
\begin{equation} \label{Eq:Eq20}
\gamma_k = \dfrac{\rho_u \eta_k \Big(\sum\limits_{m=1}^{M} \beta_{mk}\Big)^2}{ \sum\limits_{m=1}^{M} \!\left(\!\rho_u \sum\limits_{k^\prime=1}^{K} \eta_{k^\prime} \beta_{m{k^\prime}}  + \delta^2_{m}+ D_m \!\right)\!  \beta_{mk}}.
\end{equation}
\end{theorem}
\begin{proof} The proof is given in Appendix \ref{App:App1}.
\end{proof}

\section{Fronthaul Link Design} \label{Sec:Sec4}
As mentioned in the system model, it is assumed that each AP is connected to the CPU via a fronthaul link, which could be a fiber link or a free space optical link.  Although the fiber optic link has a higher channel capacity and smaller path loss compared to that of the FSO one, the deployment cost to develop the fiber link is much higher than that of the FSO one. The deployment cost refers to a sum of costs for optical transceivers, digging, installation, and medium maintenance under fiber technology, which is more costly in urban areas. In contrast, there is no need for digging and medium maintenance in the case of FSO technology, even though it demands more repeaters in urban areas. Thus, there is a trade-off between supporting higher data rates and lower deployment costs, which depends on the network's key values. Network providers should choose whether to design a system with the highest achievable data transmission rate or a system with a lower data rate but with a lower cost. On the other hand, it is also worth mentioning that due to the quantization noise added at the AP because of compression, it is not clear how large must be the capacity of the fronthaul link. Therefore in the following, we present a metric not only to consider the capacity of the fronthaul links but also somehow take into account the deployment cost as well to find out how many fiber and/or FSO fronthaul links must be used to maximize the energy efficiency of the system.

In the following, the energy efficiency of the system is presented, and the relevant optimization problem is also discussed, which includes the total power consumption of the network and fronthaul network's deployment costs.

\subsection{Energy Efficiency}
The energy efficiency of the network is defined as 
\begin{equation}\label{Eq:Eq21}
\text{EE} = \dfrac{B_s \sum\limits_{k=1}^{K}R_{u,k}}{P_{\normalfont\text{net}}+ \Omega_{\normalfont\text{fh}}} \text{~~~[bit/Joule]},
\end{equation}
where $B_s$ is the system bandwidth and to model the fronthaul network's deployment costs,  $\Omega_{\normalfont\text{fh}} = \sum\limits_{m=1}^{M} C_m \mu_m$ is used  where $\mu_m$  $[\text{Watt/bps/Hz}]$ represents the cost coefficient of data transmission over the fronthaul link of the $m$th AP which depends on the type of fonthaul link. Also, $P_{\normalfont\text{net}}$ represents the total power consumption of the network \cite{masoumi2019performance,ngo2017total}, which is given by \cite{agheli2020performance}
\begin{equation}\label{Eq:Eq22}
P_{\normalfont\text{net}}=\sum\limits_{k=1}^{K}P_k + \sum\limits_{m=1}^{M}P_{m} + B_s \sum\limits_{m=1}^{M}{C_m}{P_{\text{fh},m}} +  \sum\limits_{m=1}^{M}{P_{0,m}}\,,
\end{equation}
where $C_m$ denotes the fronthaul link capacity of $m$th AP, and
\begin{itemize}
    \item $P_k=\rho_u\eta_k$ represents the uplink transmit power of $k$th UE,
    \item $P_{m}$ denotes the sum of the consumption power of circuit elements and the signal amplification at $m$th AP,
    \item $P_{\text{fh},m}$ $[\text{Watt/Gbps}]$ is the traffic-dependent data transmission power for $m$th fronthaul link connecting AP to the CPU,
    \item $P_{0,m}$ represents the constant traffic-independent power consumed by the fronthaul link of $m$th AP.
\end{itemize}
It is assumed that a fiber-based fronthaul link has a higher deployment cost and consumes less traffic-dependent data transmission power than the FSO one, i.e. $\mu_{m_{\normalfont\text{FSO}}}\leq \mu_{m_{\normalfont\text{OF}}}$ and $P_{fh,m_{\normalfont\text{FSO}}}\geq P_{fh,m_{\normalfont\text{OF}}}$ for $m_{\normalfont\text{FSO}}=1,2,..., M_{\normalfont\text{FSO}}$ and $m_{\normalfont\text{OF}}=M_{\normalfont\text{FSO}}+1,M_{\normalfont\text{FSO}}+2,..., M$.
\subsection{Optimization Problem}
Now, we focus on maximizing the network energy efficiency subject to having a total number of $M$ APs to find out how many of the fronthaul links must be fiber or FSO, and how much larger must be the capacity of the fiber link. Thus, the optimization problem is given by
\begin{equation}\label{Eq:Eq23}
\mathcal{P}_1:\begin{cases}
 \underset{\substack{ M_{\normalfont\text{FSO}} \geq 0,\,  M_{\normalfont\text{OF}}\geq 0 ,\, N \geq 0}}{\text{max}} \text{EE}\\
 \hspace{0.8cm} \text{s.t.:}\hspace{0.4cm} M_{\normalfont\text{FSO}}+M_{\normalfont\text{OF}}\leq M.
\end{cases}
\end{equation}

By replacing  $M_{\normalfont\text{FSO}}$ with $M-M_{\normalfont\text{OF}}$, and plugging $R_{uk}$ and $P_{net}$, i.e. equations (\ref{Eq:Eq18})--(\ref{Eq:Eq20}) and (\ref{Eq:Eq22}), in $\mathcal{P}_1$, we have (\ref{Eq:Eq25}), represented at the top of the next page.
\begin{figure*}[!t]
\normalsize
\setcounter{mytempeqncnt2}{\value{equation}}
\setcounter{equation}{16}
\begin{equation}\label{Eq:Eq25}
\mathcal{P}_2:\underset{0\leq M_{\normalfont\text{OF}}\leq M,\, 0 \leq N}{\text{max}}\, \dfrac{B_s \sum\limits_{k=1}^{K}\log_2\!\left(\!1+\dfrac{\rho_u \eta_k \Big(\sum\limits_{m=1}^{M} \beta_{mk}\Big)^2}{\sum\limits_{m=1}^{M} \!\left(\!\rho_u \sum\limits_{k^\prime=1}^{K} \eta_{k^\prime} \beta_{m{k^\prime}}  + \delta^2_{m}+ D_m \!\right)\!  \beta_{mk}}\!\right)\!}{\sum\limits_{k=1}^{K}\rho_u\eta_k + \sum\limits_{m=1}^{M}P_{m} + B_s \sum\limits_{m=1}^{M}{C_m}{P_{\text{fh},m}} +  \sum\limits_{m=1}^{M}{P_{0,m}}+ \sum\limits_{m=1}^{M} C_m \mu_m}.
\end{equation}
\setcounter{equation}{\value{mytempeqncnt2}}
\hrulefill
\vspace*{4pt}
\end{figure*}
\setcounter{equation}{17}

Without loss of generality and for the sake of simplicity, we assume that all the access links experience the same large-scale fading, i.e. $\beta_{m,k}=\beta$. Also, to simplify the optimization problem, it is assumed that the users have equal transmit power, i.e., $\eta_k=\eta$, the FSO-based access points have the same power and cost parameters such that $\mu_{m_{\normalfont\text{FSO}}}=\mu_{\normalfont\text{FSO}}$ and $P_{fh,m_{\normalfont\text{FSO}}}=P_{\text{fh},\text{FSO}}$ and similarly for the fiber-based APs, we have  $\mu_{m_{\normalfont\text{OF}}}=\mu_{\normalfont\text{OF}}$ and $P_{fh,m_{\normalfont\text{OF}}}=P_{\text{fh},\text{OF}}$. Moreover, the variance of the additive Gaussian noise at all the access points are the same, i.e. $\delta^2_m=\delta^2$.  Therefore, the optimization problem is simplified as (\ref{Eq:Eq26}), shown at the top of the next page, where the parameters are presented in Table \ref{Table:Tab1}.
\begin{figure*}[!t]
\normalsize
\setcounter{mytempeqncnt}{\value{equation}}
\setcounter{equation}{17}
\begin{equation}\label{Eq:Eq26}
\mathcal{P}_3:\underset{0 \leq M_{\normalfont\text{OF}} \leq M,\, N\geq 0}{\text{max}}\, \dfrac{K B_s \log_2\!\left(\!1+\dfrac{L_{1}}{L_{2}+(M-M_{\normalfont\text{OF}})\alpha_{\normalfont\text{FSO}} + M_{\normalfont\text{OF}}\dfrac{\alpha_{\normalfont\text{OF}}}{2^{NC_{\normalfont\text{FSO}}}-1}}\!\right)\!}{\Gamma_{ep} + (M-M_{\normalfont\text{OF}}) \Gamma_{\normalfont\text{FSO}} + N M_{\normalfont\text{OF}}\Gamma_{\normalfont\text{OF}}}.
\end{equation}
\setcounter{equation}{\value{mytempeqncnt}}
\hrulefill
\vspace*{4pt}
\end{figure*}
\setcounter{equation}{18}

\noindent
\begin{table}[!t]
\begin{center}
\caption{Optimization  parameters.}\label{Table:Tab1}
\begin{tabular}{ c l } 
 \hline 
  \rowcolor{yellow}
 \small Parameter & \multicolumn{1}{|c|}{\small \text{Formula}} \\
 \hline
  \rule{0pt}{20pt} \small $L_{1}$ & \small $ M^2 \rho_u \eta \beta^2$ \\ 
  \rule{0pt}{20pt} \small $L_{2}$ & \small $M K \rho_u \eta \beta^2 + M \delta^2 \beta$ \\
  \rule{0pt}{20pt} \small $\alpha_{\normalfont\text{FSO}}$ & \small $\dfrac{\big(K \rho_u \eta \beta + \delta^2 \big )}{2^{C_{\normalfont\text{FSO}}}-1}\beta $\\
  \rule{0pt}{20pt} \small $\alpha_{\normalfont\text{OF}}$ & \small ${\big(K \rho_u \eta \beta + \delta^2 \big )}\beta$\\
  \rule{0pt}{20pt} \small $\Gamma_{ep}$ & \small $K \rho_u\eta + \sum\limits_{m=1}^{M}(P_{m}+P_{0,m})$\\
  \rule{0pt}{20pt} \small $\Gamma_{\normalfont\text{FSO}}$ & \small ${C_{\normalfont\text{FSO}}}\big(B_s{P_{fh,{FSO}}} + \mu_{\normalfont\text{FSO}}\big)$\\
  \rule{0pt}{20pt} \small $\Gamma_{\normalfont\text{OF}}$ & \small ${C_{\normalfont\text{FSO}}}\big(B_s{P_{fh,{OF}}} + \mu_{\normalfont\text{OF}}\big)$\\
  &\\
 \hline
\end{tabular}
\medskip
\end{center}
\end{table}


In the following propositions, optimal solutions of the optimization problem $\mathcal{P}_3$ are presented.

\begin{proposition} For a fixed  $M_{OF}$, the optimal capacity coefficient of the fiber compared to that of the FSO is 
\begin{align}\label{Eq:Eq27} \nonumber
&N^* \simeq\\
&\dfrac{-1}{C_{\normalfont\text{FSO}}} \log_2 \!\left(\!-\Gamma_{\normalfont\text{OF}} \alpha_{\normalfont\text{OF}} M_{\normalfont\text{OF}} \lambda_1 + \dfrac{\sqrt{\lambda_1^2-4(1-\lambda_3)\log_2\big(\dfrac{\lambda_2}{L_1}\big)}}{2.885\big(1-\dfrac{1}{\lambda_3}\big)}\right)\!\!,
\end{align}
where
\begin{align*}
\lambda_1 &= 2.443 + \log_2 \big(\dfrac{\lambda_2}{L_1}\big) + \dfrac{\lambda_4 C_{\normalfont\text{FSO}}}{\Gamma_{\normalfont\text{OF}} M_{\normalfont\text{OF}}},\\
\lambda_2 &= L_2 + (M-M_{\normalfont\text{OF}})\alpha_{\normalfont\text{FSO}},\\
\vspace{0.2cm}
\lambda_3 &= \dfrac{\lambda_2}{\alpha_{\normalfont\text{OF}} M_{\normalfont\text{OF}} \ln{(2)}},\\
\lambda_4 &= \Gamma_{ep} + (M-M_{\normalfont\text{OF}})\Gamma_{\normalfont\text{FSO}}.
\end{align*}
\end{proposition}
\begin{proof} The proof is given in Appendix \ref{App:App2}.
\end{proof}

\begin{proposition} For a fixed $N$, the optimal number of fiber-based access points is 
\begin{equation}\label{Eq:Eq28}
M_{\normalfont\text{OF}}^* \simeq  \max \bigg\lbrace  0,\dfrac{\kappa_1 \kappa_4 - \kappa_2 \kappa_3}{2 \kappa_2 \kappa_4}\bigg\rbrace  ,
\end{equation}
where $\kappa_1 = L_2 + M \alpha_{\normalfont\text{OF}}$,~~
$\kappa_2 = \alpha_{\normalfont\text{FSO}}-\dfrac{\alpha_{\normalfont\text{OF}}}{2^{NC_{\normalfont\text{FSO}}}-1}$,~~
$\kappa_3 = \Gamma_{ep} + M \Gamma_{\normalfont\text{FSO}}$, and 
$\kappa_4 = N\Gamma_{\normalfont\text{OF}} - \Gamma_{\normalfont\text{FSO}}$.
\end{proposition}
\begin{proof} The proof is given in Appendix \ref{App:App3}.
\end{proof}

As given in (\ref{Eq:Eq28}), the optimal $M_{\normalfont\text{OF}}^*$ is proportional to the number of total access points $M$. Hence, by increasing the number of APs, $M_{\normalfont\text{OF}}^*$ increases since $\kappa_1$ becomes larger faster than $\kappa_3$. Furthermore, by increasing the power and cost parameters of OF-based APs,  $M_{\normalfont\text{OF}}^*$ decreases. In contrast, by increasing the power and cost parameters of FSO-based APs,  $M_{\normalfont\text{OF}}^*$ increases. Besides, increasing the circuit power and fronthaul constant power consumption of the APs makes $M_{\normalfont\text{OF}}^*$ smaller.

\noindent
\begin{table*}[!t]
\begin{center}
\caption{Network parameters for numerical results.}\label{Table:Tab2}
\subfloat{
\hspace{-0.5cm}
\begin{tabular}{ | l | c | l || l | c | l |}
\hline
 \rowcolor{yellow}
\multicolumn{1}{|c|}{\small \text{Parameter}} & \small \text{Symbol} & \multicolumn{1}{c||}{\small \text{Value}} & \multicolumn{1}{c|}{\small \text{Parameter}} & \small \text{Symbol} & \multicolumn{1}{c|}{\small \text{Value}}\\
\hline
\small Access radio frequency & \small $f$ & \small $1.9$\, $[\text{GHz}]$ & \small Circuit power at $m$th AP & \small $P_m$ & \small $0.2$\, $[\text{Watt}]$\\
\hline
\small System bandwidth & \small $B_s$ & \small $20$\, $[\text{MHz}]$ & \small Fronthaul's constant power of  $m$th AP & \small $P_{0,m}$ & \small $0.825$\, $[\text{Watt}]$\\
\hline
\small Antenna height of AP & \small $h_{\text{AP}}$ & \small $15$\, $[\text{m}]$ & \small FSO fronthaul's traffic-dependent power & \small $P_{\text{fh},\text{FSO}}$ &\small $0.3$\, $[\text{Watt/Gbps}]$\\
\hline
\small Antenna height of UE & \small $h_{\text{UE}}$ & \small $1.65$\, $[\text{m}]$  & \small OF fronthaul's traffic-dependent power & \small $P_{\text{fh},\text{OF}}$ &\small $0.25$\, $[\text{Watt/Gbps}]$\\
\hline
\small Path-loss model minimum distance & \small $d_0$ & \small $10$\, $[\text{m}]$ & \small FSO fronthaul's cost coefficient & \small $\mu_{\normalfont\text{FSO}}$ &\small $0.003$\, $[\text{Watt/bps/Hz}]$\\
\hline
\small Path-loss model maximum distance & \small $d_1$ & \small $50$\, $[\text{m}]$ & \small OF fronthaul's cost coefficient & \small $\mu_{\normalfont\text{OF}}$ &\small $0.03$\, $[\text{Watt/bps/Hz}]$\\
\hline
\small Shadowing standard deviation & \small $\sigma_{\text{sh}}$ & \small $8$\, $[\text{dB}]$ & \small Noise power at each AP & \small $\delta^2$ & \small $k_B\!\cdot\! T_0\!\cdot\! B_s\!\cdot\! NF$\\
\hline
\small Shadowing correlation coefficient & \small $\vartheta$ & \small $0.5$ & \small Boltzmann constant & \small $k_B$ & \small $1.381\!\times\! 10^{-23}$\, $[\text{dB}]$\\
\hline
\small  UE's maximum transmission power & \small $\rho_u$ & \small $100$\, $[\text{mWatt}]$ & \small Noise temperature & \small $T_0$ & \small $290$\, $[\text{Kelvin}]$\\
\hline
\small UE's power control coefficient & \small $\eta$ & \small $0.5$ & \small Noise figure & \small $NF$ & \small $9$\, $[\text{dB}]$\\
\hline
\small  \small Capacity of FSO-based fronthaul & \small $C_{\normalfont\text{FSO}}$ & \small $2\,[\text{bps/Hz}]$ &  & & \\
\hline
\end{tabular}}
\medskip
\end{center}
\end{table*}

\section{Numerical Results and Discussions} \label{Sec:Sec5}
In this section, we present numerical results to highlight the performance of the proposed network architecture and analyse the effects of the fronthaul links on the energy and spectral efficiency of the system, numerical results are presented and discussed. Conventionally, it is assumed that $M\!=\!100$ APs and $K\!=\!10$ UEs are uniformly and randomly distributed within area of $D\!=\!1\times1$ $[\text{km}^2]$. Besides, large-scale fading is modeled by a combination of both path-loss and shadowing. Thus, it is presented by \cite{ngo2017cell}
\begin{equation} \label{Eq:Eq2}
{\beta}_{mk}=10^{\frac{{\text{PL}}_{mk}}{10}}10^{\frac{{\sigma}_{\text{sh}}z_{mk}}{10}},
\end{equation}
where PL$_{mk}\, [\text{dB}]$ represents the path-loss and $10^{\frac{{\sigma}_{\text{sh}}z_{mk}}{10}}$ denotes the shadowing with standard deviation ${\sigma}_{\text{sh}}$ while $z_{mk}$ is the shadowing correlation factor.

\subsubsection{Path-Loss Model}We employ the conventional three-slope propagation model for path-loss which is defined as \cite{ngo2017cell,masoumi2019performance}
\begin{equation} \label{Eq:Eq3}
\text{PL}_{mk}= \!\begin{cases}
-L-35 {\log}_{10}(d_{mk}) ,& d_{mk}\!>\!d_1\\
-L-15 {\log}_{10}(d_{1}) -20 {\log}_{10}(d_{mk}) ,\!\!\!\!& d_0\!<\!d_{mk}\!\leq\! d_1\\
-L-15 {\log}_{10}(d_{1}) -20 {\log}_{10}(d_{0}) ,& d_{mk} \!\leq\! d_0
\end{cases}
\end{equation}
where $d_0$ and $d_1$ are two distance measuring references, $d_{mk}$ represents the distance between $m$th AP and $k$th UE, and
\begin{align} \label{Eq:Eq4} \nonumber
&\!\!\!L  = 46.3+33.9 \, {\log}_{10}(f)-13.82\,     {\log}_{10}(h_{\text{AP}})\\
&\!- (1.1\, {\log}_{10}(f)-0.7)h_{\text{UE}}+(1.56\, {\log}_{10}(f)-0.8)\: [\text{dB}],
\end{align}
where $f\, [\text{MHz}]$ denotes access radio frequency, $h_{\text{AP}}\, [\text{m}]$ and $h_{\text{UE}}\, [\text{m}]$ are heights of AP and UE, respectively.

\subsubsection{Shadowing Model}
Because of the short distance between APs and UEs in the cell-free network, there may exist the same obstacles around the APs and the UEs, which affects the shadowing characteristics. Therefore, the correlated shadowing model is suggested in \cite{ngo2017cell}, and the shadowing correlation represented by two parameters \cite{ngo2017cell}
\begin{equation} \label{Eq:Eq5}
z_{mk}=\sqrt{\vartheta}a_m + \sqrt{1-\vartheta}b_k,
\end{equation}
where $\vartheta \!\in\! [0,1]$, $a_m\!\sim\! \mathcal{N}(0,1)$ and $b_k\!\sim\! \mathcal{N}(0,1)$ are independent random variables which model contribution to the shadowing due to the obstacles near to APs and the ones close to the UEs, respectively. 
If one sets $\vartheta\!=\!0$, the shadowing caused by each UE becomes equal at all APs, and for $\vartheta\!=\!1$, the shadowing caused by each AP becomes similar for all UEs. The parameters used for the numerical results are presented in Table \ref{Table:Tab2}; otherwise, they are clearly mentioned in the paper. 

The network energy efficiency versus $M_{\normalfont\text{OF}}$ and $N$ is depicted in Fig.~\ref{Fig:Fig.2}. It is shown that, at large values of $M_{\normalfont\text{OF}}$, increasing $N$ improves energy efficiency and then continuously decreases. For instance at $\mu_{OF}\!=\!0.03$ and $\mu_{FSO}\!=\!0.003$, the energy efficiency has a global optimum at  $M_{\normalfont\text{OF}}^*\!=\!48$ and $N^*\!=\!2$. Also, for the small values of $M_{\normalfont\text{OF}}$, the same behaviour is observed with slower changes. Besides, to study the effect of deployment cost on the network's energy efficiency, we have compared different sets of cost parameters for OF- and FSO-based fronthaul links. As expected, if it is possible to reduce the cost of fiber links, it would be better to use more fibers than the FSO links. Thus, optimizing the network performance highly depends on efficiently selecting the fronthaul links.

\begin{figure}[t!]
\centering
\pstool[scale=0.55]{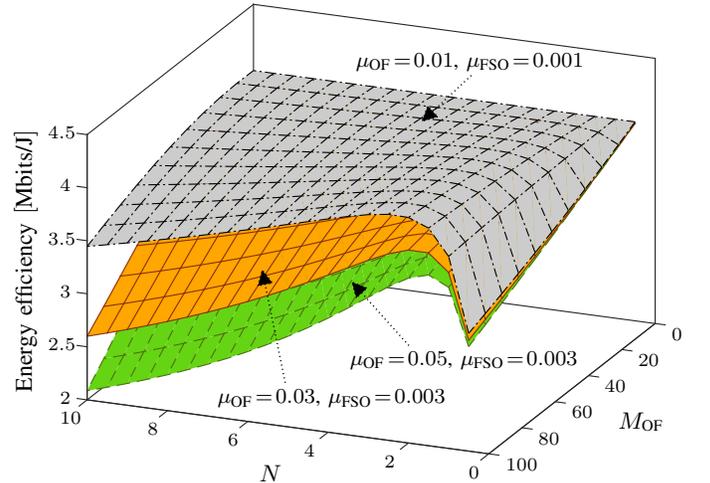}{
\psfrag{SEE}{\hspace{0.625cm}\small Energy efficiency $[\text{Mbits/J}]$}
\psfrag{M2}{\small $M_{\normalfont\text{OF}}$}
\psfrag{N}{\small $N$}
\psfrag{A}{\hspace{1cm}\footnotesize $\mu_{\normalfont\text{OF}}\!=\!0.01$, $\mu_{\normalfont\text{FSO}}\!=\!0.001$}
\psfrag{C}{\footnotesize $\mu_{\normalfont\text{OF}}\!=\!0.03$, $\mu_{\normalfont\text{FSO}}\!=\!0.003$}
\psfrag{B}{\hspace{0.5cm}\footnotesize $\mu_{\normalfont\text{OF}}\!=\!0.05$, $\mu_{\normalfont\text{FSO}}\!=\!0.003$}
\psfrag{0}{\scriptsize $0$}
\psfrag{20}{\scriptsize $20$}
\psfrag{40}{\scriptsize $40$}
\psfrag{60}{\scriptsize $60$}
\psfrag{80}{\scriptsize $80$}
\psfrag{100}{\scriptsize $100$}
\psfrag{2.5}{\hspace{-0.05cm}\scriptsize $2.5$}
\psfrag{3}{\hspace{-0.05cm}\scriptsize $3$}
\psfrag{3.5}{\hspace{-0.05cm}\scriptsize $3.5$}
\psfrag{4.5}{\hspace{-0.05cm}\scriptsize $4.5$}
\psfrag{4}{\hspace{-0.05cm}\scriptsize $4$}
\psfrag{2}{\hspace{-0.05cm}\scriptsize $2$}
\psfrag{6}{\scriptsize $6$}
\psfrag{8}{\scriptsize $8$}
\psfrag{10}{\scriptsize $10$}}
\caption{Impact of the different values of cost parameters on energy efficiency of the network for $K=10$, $M=100$, and $C_{\normalfont\text{FSO}}=2\, [\text{bps/Hz}]$.}
\label{Fig:Fig.2}
\end{figure}

The energy efficiency versus $M_{\normalfont\text{OF}}$ is represented for different values of $N$, i.e. different fiber capacity, in Fig.~\ref{Fig:Fig.3}. By increasing the value of $N$, the optimal point of $M_{\normalfont\text{OF}}$ decreases; For instance at $N \!\geq\! 8$, using FSO links for all the fronthauls is optimal, i.e.  $M_{\normalfont\text{OF}}^*\!=\!0$. Moreover, when the capacity limit of both FSO and OF links are the same, i.e., $N\!=\!1$, the same scenario occurs because the deployment cost of an FSO link is much smaller than that of an OF link. On the other hand, to maximize energy efficiency, it is optimal to use $N\!=\!2$ and $M_{OF}\!=\!48$.

\begin{figure}[t!]
\centering
\pstool[scale=0.6]{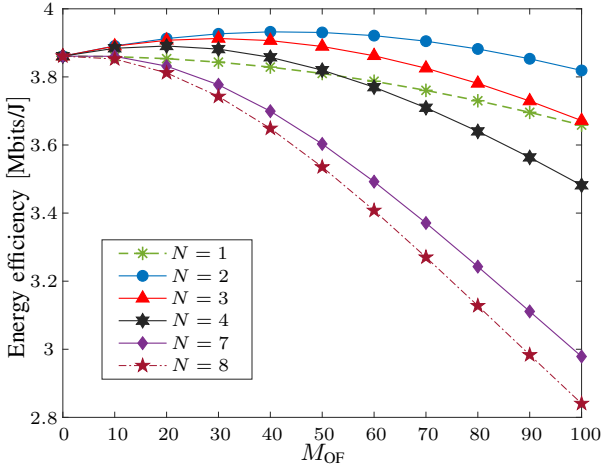}{
\psfrag{A}{\scriptsize $N = 1$}
\psfrag{B}{\scriptsize $N = 2$}
\psfrag{C}{\scriptsize $N = 3$}
\psfrag{D}{\scriptsize $N = 4$}
\psfrag{E}{\scriptsize $N = 7$}
\psfrag{F1234567}{\scriptsize $N = 8$}
\psfrag{SEE}{\hspace{-1.35cm}\small Energy efficiency $[\text{Mbits/J}]$}
\psfrag{M2}{\hspace{-0.05cm}\small $M_{\normalfont\text{OF}}$}
\psfrag{0}{\scriptsize $0$}
\psfrag{10}{\scriptsize $10$}
\psfrag{20}{\scriptsize $20$}
\psfrag{30}{\scriptsize $30$}
\psfrag{40}{\scriptsize $40$}
\psfrag{50}{\scriptsize $50$}
\psfrag{60}{\scriptsize $60$}
\psfrag{70}{\scriptsize $70$}
\psfrag{80}{\scriptsize $80$}
\psfrag{90}{\scriptsize $90$}
\psfrag{100}{\scriptsize $100$}
\psfrag{2.8}{\hspace{-0.05cm}\scriptsize $2.8$}
\psfrag{3}{\hspace{-0.05cm}\scriptsize $3$}
\psfrag{3.2}{\hspace{-0.05cm}\scriptsize $3.2$}
\psfrag{3.4}{\hspace{-0.05cm}\scriptsize $3.4$}
\psfrag{3.6}{\hspace{-0.05cm}\scriptsize $3.6$}
\psfrag{3.8}{\hspace{-0.05cm}\scriptsize $3.8$}
\psfrag{4}{\hspace{-0.05cm}\scriptsize $4$}}
\caption{Energy efficiency of the system versus $M_{\normalfont\text{OF}}$ for $K=10$, $M=100$, and $C_{\normalfont\text{FSO}}=2\, [\text{bps/Hz}]$.}
\label{Fig:Fig.3}
\end{figure}

\textcolor{black}{The cumulative  distribution functions of uplink \textit{sum}-rate and uplink \textit{per-user} rate are presented in Fig.~\ref{Fig:Fig.4}~(a) and Fig.~\ref{Fig:Fig.4}~(b), respectively. Here, we have compared the network performance for different capacity coefficients and the relevant optimal number of OF-based access points. It is shown that the network implemented based on the globally optimum values of $N$ and $M_{\text{OF}}$ provides the highest uplink rates in comparison to the other sub-optimal setups.}

\begin{figure}[t!]
\begin{minipage}{0.49 \textwidth}
\pstool[scale=0.6]{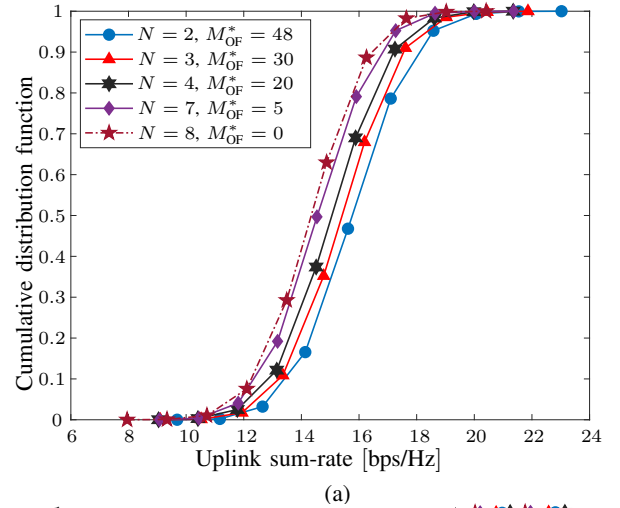}{
\psfrag{A}{\scriptsize $N = 2$, $M_{\normalfont\text{OF}}^*=48$}
\psfrag{B}{\scriptsize $N = 3$, $M_{\normalfont\text{OF}}^*=30$}
\psfrag{C}{\scriptsize $N = 4$, $M_{\normalfont\text{OF}}^*=20$}
\psfrag{D}{\scriptsize $N = 7$, $M_{\normalfont\text{OF}}^*=5$}
\psfrag{E12345678910111213141}{\scriptsize $N = 8$, $M_{\normalfont\text{OF}}^*=0$}
\psfrag{CDF}{\hspace{-2.1cm}\small Cumulative distribution function}
\psfrag{UplinkSumRate}{\hspace{-1cm}\small Uplink sum-rate $[\text{bps/Hz}]$}
\psfrag{(a)}{\hspace{0cm}\small (a)}
\psfrag{6}{\scriptsize $6$}
\psfrag{8}{\scriptsize $8$}
\psfrag{10}{\scriptsize $10$}
\psfrag{12}{\scriptsize $12$}
\psfrag{14}{\scriptsize $14$}
\psfrag{16}{\scriptsize $16$}
\psfrag{18}{\scriptsize $18$}
\psfrag{20}{\scriptsize $20$}
\psfrag{22}{\scriptsize $22$}
\psfrag{24}{\scriptsize $24$}
\psfrag{0}{\hspace{-0.05cm}\scriptsize $0$}
\psfrag{0.1}{\hspace{-0.05cm}\scriptsize $0.1$}
\psfrag{0.2}{\hspace{-0.05cm}\scriptsize $0.2$}
\psfrag{0.3}{\hspace{-0.05cm}\scriptsize $0.3$}
\psfrag{0.4}{\hspace{-0.05cm}\scriptsize $0.4$}
\psfrag{0.5}{\hspace{-0.05cm}\scriptsize $0.5$}
\psfrag{0.6}{\hspace{-0.05cm}\scriptsize $0.6$}
\psfrag{0.7}{\hspace{-0.05cm}\scriptsize $0.7$}
\psfrag{0.8}{\hspace{-0.05cm}\scriptsize $0.8$}
\psfrag{0.9}{\hspace{-0.05cm}\scriptsize $0.9$}
\psfrag{1}{\hspace{-0.05cm}\scriptsize $1$}}
\end{minipage}
\begin{minipage}{0.49 \textwidth}
\pstool[scale=0.6]{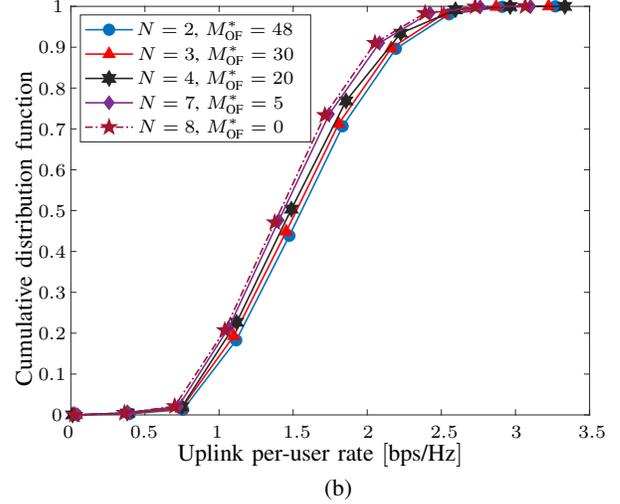}{
\psfrag{A}{\scriptsize $N = 2$, $M_{\normalfont\text{OF}}^*=48$}
\psfrag{B}{\scriptsize $N = 3$, $M_{\normalfont\text{OF}}^*=30$}
\psfrag{C}{\scriptsize $N = 4$, $M_{\normalfont\text{OF}}^*=20$}
\psfrag{D}{\scriptsize $N = 7$, $M_{\normalfont\text{OF}}^*=5$}
\psfrag{E12345678910111213141}{\scriptsize $N = 8$, $M_{\normalfont\text{OF}}^*=0$}
\psfrag{CDF}{\hspace{-2.1cm}\small Cumulative distribution function}
\psfrag{UplinkPerUserRate}{\hspace{-1.05cm}\small Uplink per-user rate $[\text{bps/Hz}]$}
\psfrag{(b)}{\hspace{0cm}\small (b)}
\psfrag{1.5}{\scriptsize $1.5$}
\psfrag{2}{\scriptsize $2$}
\psfrag{2.5}{\scriptsize $2.5$}
\psfrag{3}{\scriptsize $3$}
\psfrag{3.5}{\scriptsize $3.5$}
\psfrag{0}{\hspace{-0.05cm}\scriptsize $0$}
\psfrag{0.1}{\hspace{-0.05cm}\scriptsize $0.1$}
\psfrag{0.2}{\hspace{-0.05cm}\scriptsize $0.2$}
\psfrag{0.3}{\hspace{-0.05cm}\scriptsize $0.3$}
\psfrag{0.4}{\hspace{-0.05cm}\scriptsize $0.4$}
\psfrag{0.5}{\hspace{-0.05cm}\scriptsize $0.5$}
\psfrag{0.6}{\hspace{-0.05cm}\scriptsize $0.6$}
\psfrag{0.7}{\hspace{-0.05cm}\scriptsize $0.7$}
\psfrag{0.8}{\hspace{-0.05cm}\scriptsize $0.8$}
\psfrag{0.9}{\hspace{-0.05cm}\scriptsize $0.9$}
\psfrag{1}{\hspace{-0.05cm}\scriptsize $1$}}
\end{minipage}
\caption{CDF of (a) uplink sum-rate and (b) uplink per-user rate for $K=10$, $M=100$, and $C_{\normalfont\text{FSO}}=2\, [\text{bps/Hz}]$.}
\label{Fig:Fig.4}
\end{figure}

Fig.~\ref{Fig:Fig.5} presents the energy efficiency versus the uplink sum-rate for different values of the capacity coefficient. At a particular amount of uplink sum-rate, e.g. $22\, [\text{bps/Hz}]$, since the total power consumption of the network, given in (\ref{Eq:Eq22}), is proportional to the value of capacity coefficient, the performance degrades by increasing $N$. Moreover, to maximize the performance of the system it is not efficient to use very high-capacity fibers.
\begin{figure}[t!]
\centering
\pstool[scale=0.6]{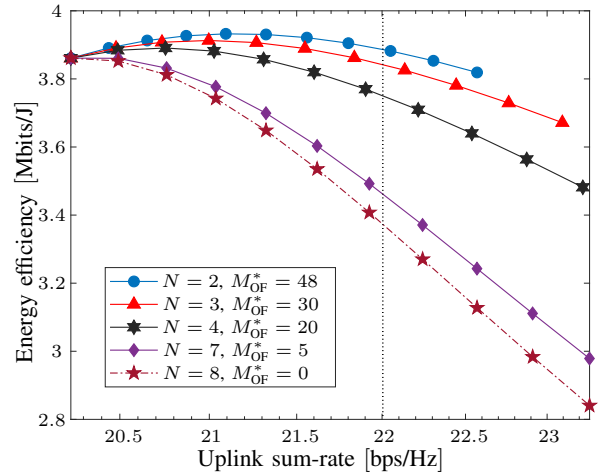}{
\psfrag{A}{\scriptsize $N = 2$, $M_{\normalfont\text{OF}}^*=48$}
\psfrag{B}{\scriptsize $N = 3$, $M_{\normalfont\text{OF}}^*=30$}
\psfrag{C}{\scriptsize $N = 4$, $M_{\normalfont\text{OF}}^*=20$}
\psfrag{D}{\scriptsize $N = 7$, $M_{\normalfont\text{OF}}^*=5$}
\psfrag{E12345678910111213141}{\scriptsize $N = 8$, $M_{\normalfont\text{OF}}^*=0$}
\psfrag{SEE}{\hspace{-1.74cm}\small Energy efficiency $[\text{Mbits/J}]$}
\psfrag{SSE}{\hspace{-1.54cm}\small Uplink sum-rate $[\text{bps/Hz}]$}
\psfrag{20.5}{\scriptsize $20.5$}
\psfrag{21}{\scriptsize $21$}
\psfrag{21.5}{\scriptsize $21.5$}
\psfrag{22}{\scriptsize $22$}
\psfrag{22.5}{\scriptsize $22.5$}
\psfrag{23}{\scriptsize $23$}
\psfrag{2.8}{\hspace{-0.05cm}\scriptsize $2.8$}
\psfrag{3}{\hspace{-0.05cm}\scriptsize $3$}
\psfrag{3.2}{\hspace{-0.05cm}\scriptsize $3.2$}
\psfrag{3.3}{\hspace{-0.05cm}\scriptsize $3.3$}
\psfrag{3.4}{\hspace{-0.05cm}\scriptsize $3.4$}
\psfrag{3.5}{\hspace{-0.05cm}\scriptsize $3.5$}
\psfrag{3.6}{\hspace{-0.05cm}\scriptsize $3.6$}
\psfrag{3.7}{\hspace{-0.05cm}\scriptsize $3.7$}
\psfrag{3.8}{\hspace{-0.05cm}\scriptsize $3.8$}
\psfrag{3.9}{\hspace{-0.05cm}\scriptsize $3.9$}
\psfrag{4}{\hspace{-0.05cm}\scriptsize $4$}}
\caption{Energy efficiency versus uplink sum-rate for $K=10$, $M=100$, and $C_{\normalfont\text{FSO}}=2\, [\text{bps/Hz}]$.}
\label{Fig:Fig.5}
\end{figure}


\section{Conclusion} \label{Sec:Sec6}
We studied the uplink of the cell-free massive MIMO network with capacity-limited fiber and FSO fronthaul links. We firstly represented the system and channel models. Next, data recovery and uplink data rates were analyzed. Eventually, an optimal fronthaul link design were proposed to maximize the network's energy efficiency, subject to the total number of the APs. With the fronthaul link design, we obtained the optimal number of APs with fiber fronthauling and each fiber link's desired capacity compared to an FSO link for a cost- and energy-efficient network. Subsequently, the number of APs with FSO fronthaul links were acquired. The network's energy and spectral efficiency were discussed through numerical results to clarify the need for optimal fronthaul allocations. The results show that the network's energy efficiency reaches its highest value by considering $N^*\!=\!2$ and $M_\text{OF}^*\!=\!48$ the optimal values. Even though higher fiber capacities provides faster data transfer, it increases the deployment cost. Thus, we have $ M_\text{OF}\!\leq\! M_\text{OF}^*$, for $N\!\geq\!N^*$, to ensure high energy and spectral efficiency with reasonable deployment cost.

\appendices
\section{SINR Derivation}\label{App:App1}
To derive the SINR, in the following after a sequence of mathematical manipulations, we derive the variances of each term given in (\ref{Eq:Eq19}).


\subsection{Computation of  $\text{DS}_k$} Since the channel coefficients are i.i.d., we obtain
\begin{equation} \label{Eq:Eq30}
\text{DS}_k=\sqrt{\rho_u\eta_{k}}\:\sum\limits_{m=1}^{M} \mathbb{E}\Big\lbrace  |g_{mk}|^2 \Big\rbrace  
=\sqrt{\rho_u\eta_{k}}\:\sum\limits_{m=1}^{M}\beta_{mk}.
\end{equation}

\subsection{Computation of $\mathbb{E}\big\lbrace  {|\text{BU}_k|}^2\big\rbrace$} As the variance of a sum of independent random variables is equal to the sum of the variances, we have 
\begin{align}\label{Eq:Eq31} \nonumber
\mathbb{E}\big\lbrace  {|\text{BU}_k|}^2\big\rbrace   &= \rho_u\eta_{k}\:\sum\limits_{m=1}^{M}\mathbb{E}\bigg\lbrace  \big| |g_{mk}|^2 - \mathbb{E}\Big\lbrace  |g_{mk}|^2\Big\rbrace  \big|^2\bigg\rbrace  \\ \nonumber
&= \rho_u\eta_{k}\:\sum\limits_{m=1}^{M}\bigg\lbrace   \mathbb{E}\Big\lbrace  \big||g_{mk}|^2\big|^2\Big\rbrace   - \big|\mathbb{E}\big\lbrace  |g_{mk}|^2\big\rbrace  \big|^2\bigg\rbrace  \\ \nonumber
&= \rho_u\eta_{k}\:\sum\limits_{m=1}^{M}\bigg\lbrace   \mathbb{E}\Big\lbrace  |g_{mk}|^4\Big\rbrace   - \beta_{mk}^2\bigg\rbrace  \\ \nonumber
&= \rho_u\eta_{k}\:\sum\limits_{m=1}^{M}\bigg\lbrace   2\beta_{mk}^2 - \beta_{mk}^2\bigg\rbrace  \\
&= \rho_u\eta_{k}\:\sum\limits_{m=1}^{M} \beta_{mk}^2.
\end{align}

\subsection{Computation of $\mathbb{E}\big\lbrace  |\text{I}_{k k^\prime}|^2\big\rbrace$} Since the channel coefficients are i.i.d., we have 
\begin{align} \label{Eq:Eq32} \nonumber
\mathbb{E}\big\lbrace  |\text{I}_{k k^\prime}|^2\big\rbrace  &= \rho_u\eta_{k^\prime}\sum\limits_{m=1}^{M}\mathbb{E}\Big\lbrace  \big|g_{m{k^\prime}}g_{mk}^*\big|^2\Big\rbrace \\ &=\rho_u\eta_{k^\prime}\sum\limits_{m=1}^{M}\beta_{m{k^\prime}}\beta_{mk}.
\end{align}

\subsection{Computation of $\mathbb{E}\big\lbrace  |\upsilon_k|^2\big\rbrace$} It can be shown that the additive Gaussian noise and the quantization noise are mutually uncorrelated, and both of them are uncorrelated with the channel coefficients. Thus,
\begin{equation} \label{Eq:Eq33}
\mathbb{E}\big\lbrace  |\upsilon_k|^2\big\rbrace  \!=\!\sum\limits_{m=1}^{M}\mathbb{E}\Big\lbrace  \big|\big(\omega_{m} + n_m \big)g_{mk}^*\big|^2\Big\rbrace  \!=\!\sum\limits_{m=1}^{M}\big(\delta^2_m+D_m\big)\beta_{mk}.
\end{equation}

By plugging (\ref{Eq:Eq30}), (\ref{Eq:Eq31}), (\ref{Eq:Eq32}), and (\ref{Eq:Eq33}) into (\ref{Eq:Eq19}), we obtain the SINR as given in (\ref{Eq:Eq20}).

\section{Optimal $N$}\label{App:App2}
To compute the optimal value of the capacity coefficient, we investigate the optimization problem $\mathcal{P}_3$. Without loss of generality, the constant parameter $KB_s$ is ignored. So, we have
\begin{equation}\label{Eq:Eq34}
\mathcal{P}_4:\underset{0 \leq M_{\normalfont\text{OF}} \leq M,\, N\geq 0}{\text{max}}\, \Lambda(N,M_{\normalfont\text{OF}}),
\end{equation}
where
\begin{align}\label{Eq:Eq35} \nonumber
&\Lambda(N,M_{\normalfont\text{OF}}) = \\
&\dfrac{ \log_2\!\left(\!1+\dfrac{L_{1}}{L_{2}+(M-M_{\normalfont\text{OF}})\alpha_{\normalfont\text{FSO}} + M_{\normalfont\text{OF}}\dfrac{\alpha_{\normalfont\text{OF}}}{2^{NC_{\normalfont\text{FSO}}}-1}}\!\right)\!}{\Gamma_{ep} + (M-M_{\normalfont\text{OF}}) \Gamma_{\normalfont\text{FSO}} + N M_{\normalfont\text{OF}}\Gamma_{\normalfont\text{OF}}}.
\end{align}
Since the fiber channel capacity is large enough, we have  $2^{NC_{\normalfont\text{OF}}}\gg1$ for the large values of $NC_{\normalfont\text{OF}}$. Hence, (\ref{Eq:Eq35}) reduces to
\begin{align}\label{Eq:Eq36}\nonumber
&\Lambda(N,M_{\normalfont\text{OF}}) = \\ \nonumber
&~~~~ \dfrac{ \log_2\left(\dfrac{L_{1}}{L_{2}+(M-M_{\normalfont\text{OF}})\alpha_{\normalfont\text{FSO}} + M_{\normalfont\text{OF}}\alpha_{\normalfont\text{OF}} 2^{-NC_{\normalfont\text{FSO}}}}\right)}{\Gamma_{ep} + (M-M_{\normalfont\text{OF}}) \Gamma_{\normalfont\text{FSO}} + N M_{\normalfont\text{OF}}\Gamma_{\normalfont\text{OF}}}\\
&= \dfrac{- \log_2\left(\dfrac{L_{2}+(M-M_{\normalfont\text{OF}})\alpha_{\normalfont\text{FSO}} + M_{\normalfont\text{OF}}\alpha_{\normalfont\text{OF}} 2^{-NC_{\normalfont\text{FSO}}}}{L_{1}}\right)}{\Gamma_{ep} + (M-M_{\normalfont\text{OF}}) \Gamma_{\normalfont\text{FSO}} + N M_{\normalfont\text{OF}}\Gamma_{\normalfont\text{OF}}}.
\end{align}

To find the optimal value of $N$ for a given $M_{\normalfont\text{OF}}$, we compute the following first derivative
\begin{align}\label{Eq:Eq37} \nonumber
&\dfrac{d}{dN}\big\lbrace  \Lambda(N,M_{\normalfont\text{OF}})\big\rbrace \overset{(a)}{=}\! -\dfrac{d}{dN}\bigg\lbrace  \dfrac{\log_2\left(\dfrac{\lambda_2 + M_{\normalfont\text{OF}}\alpha_{\normalfont\text{OF}} 2^{-NC_{\normalfont\text{FSO}}}}{L_{1}}\right)}{\lambda_4 + N M_{\normalfont\text{OF}}\Gamma_{\normalfont\text{OF}}}\bigg\rbrace\\ \nonumber
&\overset{(b)}{=} \left(\lambda_4+NM_{\normalfont\text{OF}}\Gamma_{\normalfont\text{OF}}\right)\dfrac{\alpha_{\normalfont\text{OF}} C_{\normalfont\text{FSO}} 2^{-NC_{\normalfont\text{FSO}}}}{\lambda_2 + M_{\normalfont\text{OF}}\alpha_{\normalfont\text{OF}} 2^{-NC_{\normalfont\text{FSO}}}}\\ \nonumber
&~~~+\Gamma_{\normalfont\text{OF}}\log_2(\dfrac{\lambda_2}{L_1}) + \Gamma_{\normalfont\text{OF}}\log_2\left(1+\dfrac{\alpha_{\normalfont\text{OF}}M_{\normalfont\text{OF}}}{\lambda_2} 2^{-NC_{\normalfont\text{FSO}}}\right)\\ \nonumber
&\overset{(c)}{\approx} \left(\lambda_4+NM_{\normalfont\text{OF}}\Gamma_{\normalfont\text{OF}}\right)\dfrac{\alpha_{\normalfont\text{OF}} C_{\normalfont\text{FSO}} 2^{-NC_{\normalfont\text{FSO}}}}{\lambda_2 + M_{\normalfont\text{OF}}\alpha_{\normalfont\text{OF}} 2^{-NC_{\normalfont\text{FSO}}}}\\
&~~~+\Gamma_{\normalfont\text{OF}}\log_2(\dfrac{\lambda_2}{L_1})+ \Gamma_{\normalfont\text{OF}}\dfrac{\alpha_{\normalfont\text{OF}}M_{\normalfont\text{OF}}}{\lambda_2} 2^{-NC_{\normalfont\text{FSO}}},
\end{align}
where (a) is due to replacing of the defined parameters given in (\ref{Eq:Eq27}) into (\ref{Eq:Eq36}), and (b) is valid because of differentiating and some simplifications. By use of $\ln(1+x)\approx x$ for $x \ll 1$, (c) holds due to $2^{NC_{OF}} \gg 1$.


Now, by defining $\chi := 2^{-NC_{\normalfont\text{FSO}}}$, by use of (\ref{Eq:Eq37}), we have
\begin{align}\label{Eq:Eq38} \nonumber
&\left(\lambda_4+\dfrac{\Gamma_{\normalfont\text{OF}}M_{\normalfont\text{OF}}}{C_{\normalfont\text{FSO}}\ln(2)}\left(1-\chi\right)\right)\dfrac{\alpha_{\normalfont\text{OF}} C_{\normalfont\text{FSO}} \chi}{\lambda_2 + M_{\normalfont\text{OF}}\alpha_{\normalfont\text{OF}} \chi}\\
&~~~+ \Gamma_{\normalfont\text{OF}}\log_2\big(\dfrac{\lambda_2}{L_1}\big)
+\Gamma_{\normalfont\text{OF}}\dfrac{\alpha_{\normalfont\text{OF}}M_{\normalfont\text{OF}}}{\lambda_2}\chi=0.
\end{align} 

After some mathematical simplifications, the $2$nd-order equation is rewritten as 
\begin{equation}\label{Eq:Eq39}
u_1\chi^2 + u_2\chi + u_3=0,
\end{equation}
where
\begin{itemize}
    \item $u_1:=\Gamma_{\normalfont\text{OF}}\alpha_{\normalfont\text{OF}}M_{\normalfont\text{OF}}\Big(\dfrac{\alpha_{\normalfont\text{OF}}M_{\normalfont\text{OF}}}{\lambda_{2}}-\dfrac{1}{\ln(2)}\Big)$,
    \item $u_2:=\Gamma_{\normalfont\text{OF}}\alpha_{\normalfont\text{OF}}M_{\normalfont\text{OF}}\Big( 2.443 + \log_2 \big(\dfrac{\lambda_2}{L_1}\big) + \dfrac{\lambda_4 C_{\normalfont\text{FSO}}}{\Gamma_{\normalfont\text{OF}} M_{\normalfont\text{OF}}}\Big)$,
    \item $u_3:=\Gamma_{\normalfont\text{OF}}\alpha_{\normalfont\text{OF}}M_{\normalfont\text{OF}}\Big(\dfrac{\lambda_{2}}{\alpha_{\normalfont\text{OF}}M_{\normalfont\text{OF}}}\log_2 \big(\dfrac{\lambda_2}{L_1}\big)\Big)$.
\end{itemize}

By use of the $2$nd-order equation's solution, i.e. $\chi=-u_2 - \frac{\sqrt{u_2^2-4u_1u_3}}{2u_1}$, and plugging the defined parameters given in (\ref{Eq:Eq39}) into the solution, (\ref{Eq:Eq27}) is obtained.

\section{Optimal $M_{OF}$}\label{App:App3}
Similar to Appendix \ref{App:App2}, the objective function is 
\begin{equation}\label{Eq:Eq40}
\Lambda(N,M_{\normalfont\text{OF}}) = \dfrac{ \log_2\left(1+\dfrac{L_{1}}{\kappa_1-\kappa_2M_{\normalfont\text{OF}}}\right)}{\kappa_3 + \kappa_4 M_{\normalfont\text{OF}}},
\end{equation}
where the parameters are defined in (\ref{Eq:Eq28}).

To find the optimal value of $M_{\normalfont\text{OF}}$ for a given $N$, we drive the following first derivative 
\begin{align}\label{Eq:Eq41} \nonumber
&\dfrac{d}{dM_{\normalfont\text{OF}}}\big\lbrace  \Lambda(N,M_{\normalfont\text{OF}})\big\rbrace  =\dfrac{d}{dM_{\normalfont\text{OF}}}\Bigg\lbrace  \dfrac{ \log_2\left(1+\dfrac{L_{1}}{\kappa_1-\kappa_2M_{\normalfont\text{OF}}}\right)}{\kappa_3 + \kappa_4 M_{\normalfont\text{OF}}}\Bigg\rbrace  \\ \nonumber
&= \left(\kappa_3+\kappa_4M_{\normalfont\text{OF}}\right)\dfrac{\kappa_2 L_1}{\left(\kappa_1-\kappa_2M_{\normalfont\text{OF}}\right)\left(\kappa_1-\kappa_2M_{\normalfont\text{OF}}+L_1\right)\ln(2)}\\ \nonumber
&~~~+\kappa_4\log_2\left(1-\dfrac{L_1}{\kappa_1-\kappa_2M_{\normalfont\text{OF}}+L_1}\right)\\ \nonumber
& \overset{(a)}{\approx} \left(\kappa_3+\kappa_4M_{\normalfont\text{OF}}\right)\dfrac{\kappa_2 L_1}{\left(\kappa_1-\kappa_2M_{\normalfont\text{OF}}\right)\left(\kappa_1-\kappa_2M_{\normalfont\text{OF}}+L_1\right)}\\ \nonumber
&~~~-\kappa_4\dfrac{L_1}{\kappa_1-\kappa_2M_{\normalfont\text{OF}}+L_1}\\
&=\dfrac{L_1}{\kappa_1-\kappa_2M_{\normalfont\text{OF}}+L_1}\left(\dfrac{\kappa_2\left(\kappa_3+\kappa_4M_{\normalfont\text{OF}}\right)}{\left(\kappa_1-\kappa_2M_{\normalfont\text{OF}}\right)}-\kappa_4\right)=0,
\end{align}
where (a) holds by use of $\ln(1+x)\approx x$ for $x \ll 1$. By solving (\ref{Eq:Eq41}), we obtain the solution as given in (\ref{Eq:Eq28}). Note that the derived optimal solution is correct for $\kappa_1\kappa_4 \geq \kappa_2\kappa_3 $. For large values of $N$, we have $ \kappa_1\kappa_4 < \kappa_2\kappa_3$ which results in $M_{\normalfont\text{OF}}^*=0$.


\balance
\bibliographystyle{IEEEtran}
\bibliography{References.bib}


\end{document}